\documentclass[letterpaper, 10 pt, conference]{ieeeconf}  % Comment this line out if you need a4paper

\IEEEoverridecommandlockouts                              % This command is only needed if 
\ifdefined\VersionAuthor
\else
	\makeatletter
	\AtBeginDocument{%
	\@ifpackageloaded{hyperref}
	{\def\@doi#1{\href{https://doi.org/#1}
		{\ttfamily https://doi.org/#1}\egroup}}
	{\def\@doi#1{\ttfamily https://doi.org/#1\egroup}}
	\def\doi{\bgroup\catcode`\_=12\relax\@doi}}
	\makeatother
\fi
\usepackage[ruled,vlined,linesnumbered]{algorithm2e}
	\SetKwInOut{Input}{input}
	\SetKwInOut{Output}{output}
\usepackage[utf8]{inputenc}
\usepackage[english]{babel}
\usepackage{booktabs}
\usepackage{csquotes}
 \usepackage{graphicx}
 \usepackage{amssymb,amsmath}%,amsmath,amsthm}

\usepackage[misc,geometry]{ifsym} % For \Letter

\ifdefined\VersionAuthor
	\usepackage[backend=biber,backref=true,style=alphabetic,url=false,doi=true,defernumbers=true,sorting=anyt,maxnames=99]{biblatex} % tlmgr install biblatex etoolbox logreq
	\renewbibmacro*{doi+eprint+url}{%
		\iftoggle{bbx:doi}
			{\color{black!40}\footnotesize\printfield{doi}}
			{}%
		\newunit\newblock
		\iftoggle{bbx:eprint}
			{\usebibmacro{eprint}}
			{}%
		\newunit\newblock
		\iftoggle{bbx:url}
			{\usebibmacro{url+urldate}}
			{}%
	}

\fi
\DeclareUnicodeCharacter{0301}{\'{e}}

\ifdefined\VersionWithComments
	\usepackage{draftwatermark}
	\SetWatermarkText{draft}
	\SetWatermarkScale{2}
	\SetWatermarkColor[gray]{0.9}
\fi
\usepackage[svgnames,table]{xcolor}
\definecolor{darkblue}{rgb}{0, 0, 0.7}

\usepackage[
		pdfauthor={Jawher Jerray and Laurent Fribourg},%
		pdftitle={Robust optimal control using dynamic programming and guaranteed Euler's method},
		breaklinks  = true,
		colorlinks  = true,
	\ifdefined \VersionWithComments
		pagebackref = true,
	\fi
		citecolor   = blue!50!blue,
		linkcolor   = darkblue,
		urlcolor    = blue!50!green,
	]{hyperref}

\usepackage[capitalise,english,nameinlink]{cleveref} % load after algorithm2e and hyperref
\crefname{line}{\text{line}}{\text{lines}} % to remove the capital

\usepackage{tikz}
\usetikzlibrary{decorations.pathmorphing}
\usetikzlibrary{arrows,calc}
\tikzset{
>=stealth',
help lines/.style={dashed, thick},
axis/.style={<->},
important line/.style={thick},
connection/.style={thick, dotted},
}

\ifdefined\VersionWithComments
	\usepackage[colorinlistoftodos,textsize=footnotesize]{todonotes}
\else
	\usepackage[disable]{todonotes}
\fi

\ifdefined \VersionWithComments
	\newcommand{\todoinline}[1]{\mbox{}{\color{red}{\textbf{TODO}\ifx#1\\\else:\ \fi #1}}} % here, ``\\'' stands for ``empty''
\else
	\newcommand{\todoinline}[1]{}
\fi

\usepackage{verbatim} % for 'comment'

\ifdefined \VersionWithComments
	\usepackage{soul}
	
	\newcommand{\reviewDelete}[1]{{\color{red}\st{#1}}}
\else
	
	\newcommand{\reviewDelete}[1]{}
\fi

\usepackage{amsthm}

\theoremstyle{plain}
\newtheorem{proposition}{Proposition}
\newtheorem{theorem}{Theorem}

\theoremstyle{definition}
\newtheorem{definition}{Definition}
\newtheorem{example}{Example}

\theoremstyle{remark}
\newtheorem{remark}{Remark}

\ifdefined \VersionWithComments
 	\definecolor{colorok}{RGB}{80,80,150}
\else
	\definecolor{colorok}{RGB}{0,0,0}
\fi

\RequirePackage{amsmath}

\usepackage{mathabx}

\usepackage{color}
\title{\LARGE \bf %From gradient flow to gradient descent via an error bound on Euler's method\\
Proving the Convergence to Limit Cycles using Periodically Decreasing Jacobian 
Matrix Measures}

\author{%
Jawher Jerray$^{1}$
\thanks{$^{1}$Jawher Jerray is with 
LTCI, T\'el\'ecom Paris, Institut Polytechnique de Paris, Sophia-Antipolis, France 
{\tt\small jawher.jerray@telecom-paris.fr}}
\and
Laurent Fribourg$^{2}$
\thanks{$^{2}$Laurent Fribourg is with University Paris-Saclay, CNRS, ENS Paris-Saclay, LMF, F-91190 Gif-Sur-Yvette, France
        {\tt\small fribourg@lsv.fr}}
}

\graphicspath{{./figures/}}

\begin{document}

\maketitle              % typeset the header of the contribution
\begin{abstract}
Methods based on ``(Jacobian) matrix measure'' 
to show the convergence of a dynamical system to a limit cycle
(LC), generally assume that the measure is negative everywhere
on the LC.
We relax this assumption by assuming that the matrix measure is
negative ``on average'' over 
one period of LC. 
Using an approximate Euler trajectory, we thus present a method
that guarantees the LC existence, and allows us to construct a basin 
of attraction.
This is illustrated on the example of the
Van der Pol system.
\end{abstract}

\section{Introduction}\label{sec:intro}
%``The Jacobian measure has long been used to provide estimates on solutions of systems of ordinary differential equations [16]–[20]. The following proposition allows us to bound the distance between trajectories in terms of their initial distance and the rate of expansion of the system given by ''\\
Consider the nonlinear dynamical system defined by
\begin{equation}
    \label{eqn:dyn}
    \dot{x}(t)=f(x(t)).
\end{equation}
with $f:\mathbb{R}^n\rightarrow\mathbb{R}^n$ and
$x(0)=x_0\in\mathbb{R}^n$ as initial condition.
%\footnote{Let $\xi_{x_0}(t)$ denote the solutions of \eqref{eqn:dyn} at time $t$with initial condition $\xi_{x_0}(0)=x_0$.}
%
Let $\xi_{x_0}(t)$ denote the solution of \eqref{eqn:dyn} at time $t$
with initial condition $\xi_{x_0}(0)=x_0$.

For a long time, methods %\footnote{citer aussi Lyapunov methods et, en particulier, Hauser-Chung 1993 et Barbot-Hauser 2013???} 
based on the notion of
(Jacobian) matrix measure (noted as $\mu_P(\cdot)$) have been used to show the convergence of a solution $\xi_{x_0}(t)$ to a stationary point
(see e.g. \cite{Soderlind06,AminzareS14,Giesl2022}). 
These methods essentially assume that a bound $c$
is known for
the measure of the (transverse) Jacobian matrix $J(x)$ with respect
to $x$ 
for any
$x$ belonging to a $\Omega$ forward invariant space.
Formally,
there exists $c\in\mathbb{R}$ such that, for all $x\in\Omega$:
%le systeme est {\em uniformement} transversalement contractif, i.e.:
\begin{equation}\label{eq:contraction}
\mu_P(J(x))\leq c.
\end{equation}
%Lewis 1949 [124], Dahlquist 1958 [35], Desoer-Vidyasagar 1975 [40] and Michel 2008 [138]. We follow the presentation in Aminzare-Sontag 2014 [3] and Vidyasagar 2002 [180], Giesl 2022} de la matrice $A\in\mathbb{R}^n\times\mathbb{R}^n$
%\footnote{$\in C^1(\mathbb{R}^p,{\cal S}_p^+)$, ${\cal S}_p^+$ ensemble des matrices symetriques definies positives.} 
%
Under these conditions, we have the following information 
about the distance between two solutions
$\xi_{x_0}(t)$ and $\xi_{y_0}(t)$ %of $\eqref{eqn:dyn}$
starting from two different initial conditions $x_0$ and $y_0$:
\begin{equation*} %\label{eq:exp}
\mbox{if } |y_0-x_0|\leq \varepsilon_0 \mbox{ then } |\xi_{y_0}(t) - \xi_{x_0}(t)| \leq \varepsilon_0e^{ct},
\end{equation*}
for all $t\geq 0$ (see, e.g., 
\cite{MaidensA15,FanM15}).

Similarly, 
methods
based on the transverse component $\mu_P^\bot(\cdot)$ of $\mu_p(\cdot)$
are used to show
the convergence of $\xi_{x_0}(t)$ towards a limit cycle~$\Gamma$
(see e.g. \cite{ManchesterCDC13}).
Equation \eqref{eq:contraction} becomes, for all $x\in\Omega$:
\begin{equation}\label{eq:contraction'}
\mu_P^\bot(J(x))\leq c.
\end{equation}
%Lewis 1949 [124], Dahlquist 1958 [35], Desoer-Vidyasagar 1975 [40] and Michel 2

If $c<0$, the system is said to be {\em contractive}:
The solutions $\xi_{y_0}(t)$ and $\xi_{x_0}(t)$  converge
asymptotically to each other.
Within the framework of periodic systems,
the transverse contractivity
leads to the existence and uniqueness of a limit cycle 
$\Gamma$ inside $\Omega$.
(see, e.g., \cite{ManchesterCDC13}).
%\footnote{La condition \eqref{eq:contraction'} doit en particulier etre verifi\'e sur le cercle limite $\Gamma$.}
%
These methods of proof of convergence consist essentially in 
finding a positive definite matrix $P$ (possibly depending on $x$)
ensuring \eqref{eq:contraction'} with $c<0$ for any point $x$. The
discovery of such a matrix can be done by solving  
a convex optimization problem
(via a linear matrix inequality 
or polynomial sum of squares, see \cite{ManchesterCDC13}), 
but such a problem of optimization
does not always have a solution.

We propose here to relax the criterion  \eqref{eq:contraction'} by allowing that
$\mu_P^\bot(J(x))$ can be locally  $\geq 0$
on $\Gamma$,
provided that $\mu_P^\bot(J(x))$ is negative on average
on $\Gamma$, i.e.:
\begin{equation}\label{eq:periodic}
\int_0^T\mu_P^\bot(J(\xi_{x^*}(t)))dt\leq c <0,
\end{equation}
where $T$ is the period of $\Gamma$, and $x^*$ a point of $\Gamma$.

We give a set-based criterion, based on the use of
an Euler trajectory $\tilde{x}(t)$ of initial condition $\tilde{x}_0$.
More precisely, we give an upperbound $\delta(t)$ of the error~$e_{y_0}(t)=|\tilde{x}(t)-\xi_{y_0}(t)|$ where the initial point $y_0$ 
of the exact solution $\xi_{y_0}(t)$ is located in the vicinity of $\tilde{x}_0(t)$.
By showing that $\delta(t)$ decreases at each round of the Euler trajectory,
we prove the contraction ``in average'' of the system, 
thus highlighting the presence of 
a limit cycle $\Gamma$ in the vicinity of $\tilde{x}(t)$
(see Theorem~\ref{lemma:new1}).  This also allows us 
to construct an invariant zone ${\cal C}$ around~$\Gamma$.

Thanks to an additional numeric criterion, which is a discrete version of~\eqref{eq:periodic} (see~\eqref{eq:newrelax}, Section~\ref{sec:attraction})
%\footnote{involving a ``smaller'' initial error $\delta'_0\in(0,\delta_0]$ and a smaller time-step $h_\ell=\frac{h}{2^{\ell}}$},
we then
determine a basin of attraction of $\Gamma$. See Theorem~\ref{th:key1}.
%
%This criterion (voir \eqref{eq:new}, Section~\ref{sec:approx})
%makes use of an Euler trajectory
%\tilde{\Gamma}$, and allows the construction of
%of a basin of attraction ${\cal C}$ of $\Gamma$. 
%
The method is illustrated on the example of the Van der Pol system.

Note that, although our method can be defined using a norm 
$|x|_P:=\sqrt{x^\top P x}$ 
for a symmetric positive definite
matrix $P$, we restrict ourselves in the following to
the case where $P$ is the identity matrix $I$, 
the norm $|\cdot|_I$ (just denoted $|\cdot|$) is
the Euclidean norm, and 
$\mu_I(\cdot)$ (just denoted $\mu(\cdot)$) is
the associated matrix measure
(see Section~\ref{sec:preliminaries}).

In summary, our contribution is to
\begin{itemize}
\item give a sufficient set-based condition (see \eqref{eq:new})
that guarantees the existence of a limit cycle $\Gamma$, 
and allows us to construct an invariant zone ${\cal C}$ around $\Gamma$,  
\item give an additional numeric condition \eqref{eq:newrelax}
(which is a discrete version of \eqref{eq:periodic})
%\footnote{together with \eqref{eq:dl}???}
that allows us to
determine a basin of attraction of $\Gamma$, %\eqref{eq:A} and \eqref{eq:d0}
\item illustrate these points on the 
Van der Pol example.
\end{itemize}
\subsection*{Plan of the paper}
After some preliminaries (Section~\ref{sec:preliminaries}), we give 
a criterion that guarantees the  existence of a limit cycle $\Gamma$
(Section~\ref{sec:approx}).
%, and  allows us to construct an invariant zone around $\Gamma$ 
We then give an additional condition 
%under the form of a strong version of \eqref{eq:periodic}
that allows us to
determine a basin of attraction of $\Gamma$
(Section~\ref{sec:attraction}).
We conclude in Section~\ref{sec:conclusion}.
\section{Preliminaries}\label{sec:preliminaries}
We denote by $\mathbb{R}$ and $\mathbb{N}$ the set of real and natural numbers, respectively. These symbols are annotated with subscripts to restrict them in the usual way, e.g., $\mathbb{R}_{\geq 0}$ denotes the non-negative real numbers.
We denote by $\mathbb{R}^n$ a $n$-dimensional Euclidean space. 
For a matrix $A$, we denote by $A^\top$ the transpose of $A$. 
The Euclidean norm is denoted by $|\cdot|$. %\footnote{Dire qu'on abrege $\mu_P$ par $\mu$.}
The ball of center $x\in\mathbb{R}^n$ and radius $\delta\in\mathbb{R}_{\geq 0}$
is denoted by $B(x,\delta)$ (i.e.,  $B(x,\delta)=\{y\in\mathbb{R}^n: |x-y|\leq \delta\}$).
The distance $d(x,\Gamma)$ of $x\in\mathbb{R}^n$ to 
$\Gamma\subset \mathbb{R}^n$ is defined as $\inf\{|x-y|: y\in\Gamma\}$.
A set of successive integers of the form $\{1,2,\dots,k\}$ is 
abbreviated as $[k]$. The scalar product of $x\in\mathbb{R}^n$ and $y\in\mathbb{R}^n$ is written $\langle x, y\rangle$.
Let
$J(x)$ be the Jacobian matrix
of the vector field $f(x)$.
Among the eigenvectors of matrix $\frac{J(x)+J^\top(x)}{2}$,
let $v_0(x)$ (resp. ${\cal E}_0(x)$) be the tangent eigenvector 
(resp. associated eigenvalue) led by $f(x)$,
and $v_1(x),\dots, v_{n-1}(x)$ 
(resp. ${\cal E}_1(x),\dots, {\cal E}_n(x)$) be the other eigenvectors
(resp. associated eigenvalues).
Let
\begin{equation}\label{eq:mu00}
%    \label{eqn:H03}
    \mu(J(x))=\max_{i=0,1,\dots,n-1}{\cal E}_i(x).%\footnote{probleme: le vecteur $v_i(x)$ ``bouge'' et n'est pas necessairement parallele \`a la tangente.}
\end{equation}
Given a global bound $\lambda$ on $\mu(J(x))$, we
know that all trajectories of \eqref{eqn:dyn} with initial conditions in $B(x_0,\delta_0)$ lie in
$B(\xi_{x_0}(t),\delta_0 e^{\lambda t})$ 
%where $x(t)$ is the solution of \eqref{eqn:dyn} with initial condition $x(0)=x_0$ 
(see, e.g., 
\cite{MaidensA15,FanM15}).
If $\lambda < 0$ then the system~\eqref{eqn:dyn} is said
to be {\em contracting} (cf. \cite{LohmillerS98}). 
Let
\begin{equation}\label{eq:mubot}
%    \label{eqn:H03}
    \mu^\bot(J(x))=\max_{i=1,\dots,n-1}{\cal E}_i(x).%\footnote{probleme: le vecteur $v_i(x)$ ``bouge'' et n'est pas necessairement parallele \`a la tangente.}
\end{equation}
Note that, while $i=0$ belongs to the index domain of \eqref{eq:mu00}, 
the index $i=0$ is discarded in \eqref{eq:mubot} because it corresponds to
the tangent direction.

Let $L$ denote the Lipschitz constant of vector field $f$,
and $M_f$ an upperbound on magnitude of $f$
(i.e:
$M_f\geq |f(z)|$ for all $z\in\Omega$). %\footnote{or $z\in\Omega$???}).

We denote by $\tilde{x}_i$ the Euler discretization of 
\eqref{eqn:dyn}
at time $t_i=ih$, for $i\in\mathbb{N}$, where $h$ is the time-step size.
Given an initial point $\tilde{x}_0\in\mathbb{R}^n$, 
$\tilde{x}_{i+1}$ is defined, for $i\in\mathbb{N}$, by: 
\begin{equation}
    \label{eqn:Euler}
    \tilde{x}_{i+1}=\tilde{x}_{i}+hf(\tilde{x}_{i}).
\end{equation}
For $s\in[0,h]$, $t=ih+s$, let $\tilde{x}_i(s) =\tilde{x}_i+sf(\tilde{x}_i)$
(so $\tilde{x}_i(0)=\tilde{x}_i$, $\tilde{x}_i(h)=\tilde{x}_{i+1}$),
and $\tilde{x}(t)=\tilde{x}_i(s)$.
%
%Let $| \cdot |$ be a norm on $\mathbb{R}^n$ and $\|\cdot\|$ be its induced norm on the set of real matrices of dimension $n\erlind06})
%The measure of a matrix is the one-sided derivative of $\|\cdot\|$ at $I\in\mathbb{R}^{n\times n}$ in the direction $A$ ($\mu(A)=\lim_{h\rightarrow 0^+}\frac{\|I+hA\|-1}{h}$where $I$ is the identity matrix).
%
%The logarithmic norm of a matrix can be calculated for different norms.
%Here we restrict ourselves to the Euclidean vector norm $|x|=\sqrt{\Sigma_j x_j^2}$. %The induced matrix norm $\|A\|=\sqrt{\max_j \Lambda_j(A^\top A)}$,
%\footnote{The setreachable at time $t$ from initial set $S\subseteq\mathbb{R}^n$ is denoted $\Reach_t(S)$ and isdefined as $\{\xi_x(t): x \in S\}$. The set reachable from the initial set$S$ over an interval $[t_1, t_2]$ is denoted $\Reach_{[t_1,t_2]}(S) = \bigcup_{t\in[t_1,t_2]} \Reach_t(S)$???}\
%
%Let us consider a solution $\tilde{x}_i$ of \eqref{eqn:Euler} and $x\in D_{\delta_i}(\tilde{x}_i)=B(\tilde{x}_i,\delta_i)\cap S_i$ with $\delta_i\in\mathbb{R}_{\geq 0}$.
%
%\begin{definition}\label{def:Dgamma}
%Let $D=\frac{1}{2LM}$,\\
%For $i\in\mathbb{N}$, let $D'_i(\gamma)=D_\gamma\min_{k\in\{0,1,\dots i\}}a_k$.
%\end{definition}

\section{Proof of existence of a limit cycle}\label{sec:approx}
We now give a sufficient condition to guarantee the existence of a
limit cycle $\Gamma$ solution of \eqref{eqn:dyn},
and
construct a forward invariant zone around $\Gamma$. %We now use a Euler trajectory $\tilde{\Gamma}$ as a reference trajectory.

%\subsection{Sufficient condition  for the }\label{sec:prelim}
%\section{Preliminaries}
%\subsection*{Contraction}
%\subsection*{Definition de $\delta_0$}
%\subsection*{Definition de $\tilde{\Gamma}$, $\tilde{x}_0$,$\delta_0$, $h$, $R_i$}
Consider an Euler trajectory $\tilde{\Gamma}$,
i.e. a solution
$\tilde{x}(t)$ of \eqref{eqn:Euler} with time-step $h\in\mathbb{R}_{>0}$
and initial condition $\tilde{x}_0\in\mathbb{R}^n$.
Let $S_0$ be the hyperplan through $\tilde{x}_0$ 
orthogonal to $f(\tilde{x}_0)$.
Let $\tilde{x}_i =\tilde{x}(ih)$, for $i\in\mathbb{N}$.
We denote by $S_i$ the plan passing through $\tilde{x}_i$ and orthogonal
to $f(\tilde{x}_i)$.
Likewise, for  $s\in[0,h)$, $S_i(s)$ is the plan   
through $\tilde{x}_i(s)$ orthogonal to $f(\tilde{x}_i(s))$.
%
%\footnote{En particulier, etant donn\'e $\delta_i$, on a pour tout $y\in B(\tilde{x}_i,\delta_i)\cap S_i$: %\begin{equation*} %\label{eq:basic6'} $f(\tilde{x}_i)^\top f(y)>0$.???}
%\end{equation*}
%

We suppose that $\tilde{x}(t)$ returns
for the first time (in the good direction)
to $S_0$ at time $t=\tilde{R}_1\in((N_1-1)h,N_1h]$ 
for some~$N_1\in\mathbb{N}_{>0}$. 
%
%$\tilde{R}_1=(N_1-1)h+\tilde{r}_1$ with $\tilde{r}_1\in (0,h]$, 
%$N=\lceil\tilde{T}\rceil/h$,
%$\tilde{x}(\tilde{R}_1)\in S_0 \mbox{ and }
%f(\tilde{x}_0)^\top  f(\tilde{x}(\tilde{R}_1))>0$.\\
%
More generally, 
we suppose that $\tilde{x}(t)$ returns
for the $p$-th time ($p\in\mathbb{N}_{>0}$)
to $S_0$ at time $t=\tilde{R}_p\in((N_p-1)h,N_ph]$ 
for some~$N_p\in\mathbb{N}_{>0}$, so we have: 
%$\tilde{R}_p=(N_p-1)h+\tilde{r}_p$ with $\tilde{r}_p\in (0,h]$, 
%$N_p=\lceil \tilde{R}_p\rceil/h$,
$$\tilde{x}(\tilde{R}_p)\in S_0 \mbox{ and }
f(\tilde{x}_0)^\top  f(\tilde{x}(\tilde{R}_p))>0.$$
By convention, let $\tilde{R}_0=0$.
Note that
$S_0$ is ``between'' $S_{N_p-1}$ and $S_{N_p}$ for all $p\in\mathbb{N}_{>0}$
(see Figure~\ref{fig:v_cycle}, Section~\ref{sec:approx}).

Given $\delta_0\in\mathbb{R}_{>0}$ and $y_0\in B(\tilde{x}_0,\delta_0)\cap S_0$,
we suppose that 
$\xi_{y_0}(t)$ returns to $S_0$ for the 1st time
at $t=T_{y_0}(1)\in\mathbb{R}_{>0}$.
So we have:
\begin{equation*} %\label{eq:basic''}
\xi_{y_0}(T_{y_0}(1))\in S_0 \mbox{ and }
f(\tilde{x}_0)^\top  f(\xi_{y_0}(T_{y_0}(1)))>0.
\end{equation*}
We will suppose that there exists $\eta\in\mathbb{R}_{>0}$
such that:
\begin{equation}\label{eq:eta}
T_{y_0}(1)\geq \eta>0 \mbox{ for all $y_0\in B(\tilde{x}_0,\delta_0)\cap S_0$}.
\end{equation}
This will be used in Theorem~\ref{lemma:new1} 
to show $T_{y_0}(p)\rightarrow\infty$
as $p\rightarrow\infty$, where $T_{y_0}(p)$ is the $p$-th return time
of $\xi_{y_0}(t)$ to $S_0$.

Along the lines of \cite{Giesl2022},
we now synchronize the time between $\tilde{x}_i(t)$ 
and the corresponding solution $\xi_{y_i}(t)$ 
of \eqref{eqn:dyn} with $y_i\in B(\tilde{x}_i,\delta_i)\cap S_i$. %\footnote{$\delta_i$???}\cap S_i$.
More precisely, we consider the time-reparametrisation
$ \theta_{y_i}(\cdot)$ defined so that
$\xi_{y_i}(\theta_{y_i}(t))$ belongs to $S_i(t)$.  %of \eqref{eqn:dyn} 
%of \eqref{eqn:dyn} and \eqref{eqn:Euler} respectively
Hence, given two adjacent points $\tilde{x}_i\in S_i$ 
and $y_i\in S_i$ with $(\tilde{x}_i-y_i)$ perpendicular
to $f(\tilde{x}_i)$, we define $\theta_{y_i}(t)$
 for all $s\in[0,h)$ in an ``implicit'' manner as follows: 
\begin{equation}\label{eq:thetai}
(\xi_{y_i}(\theta_{y_i}(s))-\tilde{x}_i(s))^\top f(\tilde{x}_i(s))=0.
\end{equation}
%\footnote{(see Figure~\ref{fig:theta}???)}.
%
This is possible due to the implicit function theorem if $h$ and
$|y_i-x_i|$  are sufficiently small (see~\cite{Giesl2022}, Section 2.1).
Using \eqref{eq:thetai}, we define 
$\theta_{y_i}(s)$ with $\theta_{y_i}(0)=y_0$,
and $y_{i+1}=\xi_{y_i}(\theta_i(h))$.
Note that $y_i\in S_i$ for all $i \in\mathbb{N}$.
We can now define $\Theta_{y_0}(\cdot)$ so that
$\xi_{y_0}(\Theta_{y_0}(t))$ is synchronized with
the solution~$\tilde{x}(t)$ of \eqref{eqn:Euler} (with
initial condition $\tilde{x}_0$). Formally,
given $y_0\in \mathbb{R}^n$,  
%
%\begin{definition}\label{def:Theta}
$\Theta_{y_0}(0) := y_0$ and $\Theta_{y_0}(t) := \theta_{y_i}(s)$
for $t=ih+s$ with  $s\in(0,h]$.
%\end{definition}
%
Note $\xi_{y_0}(\Theta_{y_0}(ih))=y_{i}\in S_{i}$
for all $i\in\mathbb{N}$.
We will ensure $\Theta_{y_0}(t)\rightarrow\infty$
as $t\rightarrow\infty$ thanks to 
assumption \eqref{eq:eta} %\footnote{qui est elle-meme garanti for $\delta_0$ suffisamment petit et $\tilde{x}_0$ suffisamment close to $\Gamma$}
coupled with assumption \eqref{eq:new} (see Theorem~\ref{lemma:new1}).

%We also assume $\dot{\theta}_i(s)>0$ for all $i\in\mathbb{N},s\in[0,h]$.\footnote{d'o\`u ca vient ???, \eqref{eq:implicit}???}

We now define $a_{i+1}(y_0)$ and $b_{i+1}(y_0)$ ($i\in\mathbb{N}$), or for the sake of
notation simplicity, just $a_{i+1}$ and $b_{i+1}$, as follows: 
\begin{definition}\label{def:ab}
%there exist  such that,for all $z\in Z_i$: %\footnote{$B(\tilde{x}_{i},\delta_i)\cap S_i$ instead of $Z_i$???}:
%Let $a_0,b_0$ such that:$0<a_0\leq \dot{\theta}_{y_0}(0)\leq b_0$,
For $i\in\mathbb{N}$, let:
\begin{equation*} %\label{eq:basic5'}
0<a_{i+1}\leq \inf_{s\in[0,h]}\dot{\theta}_{y_i}(s),\ \ \ \ \ \ \  b_{i+1}\geq \sup_{s\in[0,h]}\dot{\theta}_{y_i}(s).
\end{equation*}
\end{definition}
\begin{remark}
The existence of $a_{i+1},b_{i+1}$ for all $i\in\mathbb{N}$ follows from
the assumptions that 
the 1st return times of $\tilde{x}_0(t)$
and $\xi_{y_0}(t)$ to $S_0$ 
are {\em finite}
(i.e., $\tilde{R}_1\in\mathbb{R}_{>0}$
and $T_{y_0}(1)\in\mathbb{R}_{>0}$
for all $y_0\in B(\tilde{x}_0,\delta_0)\cap S_0$)
together
with conditions~\eqref{eq:eta}-\eqref{eq:new} (see Theorem~\ref{lemma:new1}).
The bounds $a_i$ and $b_i$ can be computed using the implicit function theorem
by time derivation of \eqref{eq:thetai}
(see \cite{Giesl2022}, Section 2.1).
The assumption that $a_{i+1}>0$  (i.e., 
$\inf_{s\in[0,h]}\dot{\theta}_{y_i}(s)>0$) is true %in the case where
when $h$ is sufficiently small and
%the existence of a limit cycle $\Gamma$ is known, 
a lowerbound $m>0$ exists on 
the magnitude of vector field~$f$ 
(i.e., $|f(x)|>m$ for all $x\in\Omega$). 
\end{remark}
\noindent
We now define $\alpha_i$, $Y_i$  and $Z_{i+1}$
along the lines of \cite{MaidensA15} (Algorithm~1), as follows:
\begin{definition}\label{def:Zi} 
\begin{itemize}
\item Let

$\alpha_0=\delta_0$,  

$\alpha_{i+1}=\alpha_{i}+b_{i+1} M_f h$ for $i\in\mathbb{N}$,\\
More generally let

$\alpha_{i}(s)=\alpha_{i}+b_{i+1} M_f s$ for $i\in\mathbb{N}, s\in[0,h]$.
\item Let
$Y_i=B(\tilde{x}_i,\alpha_i)\cap S_i$, for all $i\in\mathbb{N}$.\\
More generally let:

$Y_i(s)= B(\tilde{x}_{i}(s),\alpha_{i}(s))\cap S_{i}(s)$
for all $i\in\mathbb{N}, s\in[0,h]$.
%$Y_i(0)=B(\tilde{x}_{i},\alpha_{i})\cap S_{i}$.\\
%$Y_i(h)=B(\tilde{x}_{i+1},\alpha_{i+1})\cap S_{i+1}$.\\
%
\item Let $Z_{i+1}=\bigcup_{s\in[0,h]}Y_i(s)$ for all $i\in\mathbb{N}$.
\end{itemize}
\end{definition}
\begin{remark}
Note that $Y_0=B(\tilde{x}_0,\delta_0)\cap S_0$
and $Y_{i+1}=Z_{i+1}\cap S_{i+1}$.
Note also that, given $y_0\in Y_0$, we have:
%\begin{equation}\label{eq:Zi}
$\xi_{y_i}(\theta_{y_i}(s))\in Z_{i+1}$
%\end{equation}
for all 
$i\in\mathbb{N}, s\in[0,h]$. We have also:
$\xi_{y_0}(\Theta_{y_0}(t))\in Z_{i+1}$ for all $t=ih+s$ with $s\in[0,h]$.
This implies that $Z_{i+1}$ %($i\in\mathbb{N}$) 
is an overapproximation of~the ``reachability'' set 
$\{z\in \mathbb{R}^n|\ z=\xi_{y_0}(\Theta_{y_0}(ih+s))\ \mbox{ for some $y_0\in Y_0$ and $s\in[0,h]$}\}.$
\end{remark}
\noindent
Let us now define $\Lambda_{i+1}\in\mathbb{R}$  ($i\in\mathbb{N}$), as follows:
\begin{equation*} %\label{eq:local_contractivity'}
\Lambda_{i+1}\geq \sup_{z\in Z_{i+1}}\mu^\bot[J(z)],
\end{equation*}
\noindent
Note that~$\Lambda_{i+1}$ satisfies for all $z_1,z_2\in Z_{i+1}$ (see e.g.
\cite{Soderlind06}):
\begin{equation}\label{eq:mu0}
\langle f(z_1)-f(z_2),z_1-z_2\rangle\leq \Lambda_{i+1}|z_1-z_2|^2.
%\mu^\bot(J(y))\leq \frac{\lambda^*}{2}<0.
\end{equation}
%
%%%%%%%%%%%%%%%%%%%%%%%%%%%%%%%%%%%%%%%%%%%%%%%%%%%%
%
Let $\gamma>0$ be a positive real
(arbitrarily chosen).
%We define $\sigma_i$ and $\sigma(t)$  
\begin{definition}
%\begin{itemize}
%\item $\lambda'_i=\tilde{\rho}_i\tilde{\lambda}_i$ with
For all $i\in\mathbb{N}$, let

\hspace*{5mm}$\sigma_{i+1}=\frac{1}{2}a_{i+1} %\footnote{$(1-\beta)$ instead of $1/2$???}
\Lambda_{i+1}$ \hspace*{\fill}if $\Lambda_{i+1}<-\gamma$,

\hspace*{5mm}$\sigma_{i+1}=\frac{3}{2}b_{i+1} %\footnote{$(1+\beta)$ instead of $1/2$???}
\max(|\Lambda_{i+1}|,\gamma)$ \hspace*{\fill}if $\Lambda_{i+1}\geq -\gamma$.
%
%\hspace*{5mm}$\sigma(t) =\sigma_{i+1}$ for $s\in(0,h],t=ih+s$\footnote{On n'a pas besoin de definir $\sigma(t)$; seulement $\sigma_i$ vont servir plus tard.
%(\`a moins qu'on utilise $\int \sigma(t)dt$)???}
%\footnote{or $=\sigma_{i,i+1}$???}.
\end{definition}
\begin{remark}\label{rk:inf5}
Note that $a_i, b_i, Z_i, \Lambda_i, \sigma_i$ are {\em not} defined for $i=0$.
The constant $\sigma_{i+1}$ is a conservative approximation
of the matrix measure $\Lambda_{i+1}$ on $Z_{i+1}$.
The 
constant $\gamma$ is used 
to isolate the problematic neighborhood of
$\mu^\bot[J(z)]$ around 0,
(and get a {\em positive} lowerbound
of $|\sigma_{i+1}|$):
If $z$ is such that $|\mu^\bot[J(z)]|< \gamma$,   
then $|\mu^\bot[J(z)]|$ is replaced by~$\gamma$. 
It follows from the definition that 
for all $i\in\mathbb{N}$:
%\footnote{$\sigma_i\geq \frac{1}{2}a_i\Lambda_i\geq\frac{1}{2}\Lambda_i\min_{k\in\{i,i+1\}}a_k$???}
\begin{equation}\label{eq:minsigma}
|\sigma_{i+1}|\geq \frac{\gamma}{2}a_{i+1}>0
\end{equation}
\end{remark}
%\subsection*{Definition iterative de $\delta_i$}
Given $\delta_0$, 
let us define $\delta_i$, $\delta(t)$ and $\tilde{e}_{y_i},\tilde{e}_{y_0}(t)$
as follows:
\begin{definition}\label{eq:deltatilde}
Let $\tilde{e}_{y_{0}}=|\tilde{x}_{0}-y_{0}|$, and for $i\in\mathbb{N}$ let:

$\delta_{i+1}=\delta_ie^{\sigma_{i+1} h}$, \hspace*{\fill}$\tilde{e}_{y_{i+1}}=|\tilde{x}_{i+1}-y_{i+1}|$.\\
More generally for $i\in\mathbb{N}, s\in[0,h], t=ih+s$, let:
 
$\delta_i(s) =\delta_ie^{\sigma_{i+1} s}$ %\footnote{au lieu de $\delta_i(0)=\delta_i$ et $\delta_i(s) =\delta_ie^{\sigma_{i+1} s}$ for $s\in(0,h]$, on peut simplement definir $\delta_i(s) =\delta_ie^{\sigma_{i+1} s}$ for $s\in[0,h]$.}, 
\hspace*{\fill}$\tilde{e}_{y_{i}}(s)=|\tilde{x}_{i}(s)-y_{i}(\theta_{y_i}(s))|$,

$\delta(t)=\delta_i(s)$, \hspace*{\fill}$\tilde{e}_{y_0}(t)=|\tilde{x}(t)-\xi_{y_0}(\Theta_{y_0}(t))|$.
\end{definition}
\begin{remark}\label{rk:delta}
%Note  that $\delta(0)=\delta_0$ and
%$\delta((i+1)h)=\delta_{i+1}=\delta_0 e^{h\Sigma_{k=1}^{i+1}\sigma_{k}}$.
%
For $i\in\mathbb{N}$ and $s\in[0,h]$, we have:

$\delta(ih+s)=\delta_{i}(s)=\delta_0 e^{h\Sigma_{k=1}^{i}\sigma_{k}+s\sigma_{i+1}}$\\
if we adopt the convention
$\Sigma_{k=1}^i\sigma_k=0$ for $i=0$.
\end{remark}
\noindent
We now show that, under certain condition, $\tilde{e}_{y_0}(t)\leq \delta(t)$.
%the following counterpart of Proposition~\ref{prop:1}:
%\footnote{This can also be seen as  an adaptation of Theorem~1 of \cite{SNR17} to the context of {\em transverse} contraction.}:
%Using \eqref{eq:minsigma},we have
%\begin{equation}\label
%|\sigma_i|\geq (1-\beta)\gamma\min_{k\in[i]}a_k.
%\end{equation}
\begin{proposition}\label{prop:basicn}
Let $i\in\mathbb{N}, y_i\in Y_i$.\\
If $\tilde{e}_{y_i}\leq\delta_i$, we have:
\begin{equation}\label{eq:key2}
\tilde{e}_{y_i}(s)\leq\max(\delta_i(s),
h(\tilde{M}_{i+1}(\frac{2L}{\gamma a_{i+1}}+1)+b_{i+1}M_f))
\end{equation} 
for all $s\in[0,h]$,
where $\tilde{M}_{i+1}$ an upperbound of $|f(\tilde{x}_i(s))|$
on $s\in[0,h]$.

Furthermore, for all $i\in\mathbb{N}_{>0},s\in[0,h],y_0\in Y_0$:
\begin{equation}\label{eq:max0}
\tilde{e}_{y_{i-1}}(s)\leq \max(\delta_{i-1}(s),h(\tilde{M}_i(\frac{2L}{\gamma a_i}+1)+b_iM_f)).
\end{equation}
\end{proposition}
\begin{proof}
The proof of \eqref{eq:key2} is 
an adaptation of the proof of Theorem~1 of \cite{SNR17} to the context of
{\em transverse} contraction.

Consider $y_i\in Y_i$. 
Let: $\tilde{e}_{y_i}(t)=|\xi_{y_i}(\theta_{y_i}(t))-\tilde{x}_i(t)|$ for $t\in[0,h]$.
Hence $\tilde{e}_{y_i}(0)=|y_i-\tilde{x}_i|\leq \delta_i$.
For the sake of simplicity, we will write $\tilde{x}_t$ instead of
$\tilde{x}_i(t)$, $y_{\theta t}$ instead of $\xi_{y_i}(\theta_{y_i}(t))$, $\dot{\theta}_t$ instead of
$\dot{ \theta}_{y_i}(t)$.
%We also abbreviate $S(\Delta,\varepsilon)$ to $S$. 

For all $t\in[0,h]$, we have by~\eqref{eq:mu0}:
\begin{equation}\label{eq:mu1}
\langle f(y_{\theta t})-f(\tilde{x}_t),y_{\theta t}-\tilde{x}_t\rangle\leq \Lambda_{i+1}|y_{\theta t}-\tilde{x}_t|^2
%\mu^\bot(J(y))\leq \frac{\lambda^*}{2}<0.
\end{equation}
since $y_{\theta t},\tilde{x}_t\in Z_{i+1}$.
So for all $t\in[0,h]$, $\tilde{e}_{y_i}(t)=|y_{\theta t}-\tilde{x}_t|$ satisfies:\\
\\
$\frac{1}{2}\frac{d}{dt}(\tilde{e}_{y_i}^2(t))$
$=\langle \dot{\theta}_tf(y_{\theta t})-f(\tilde{x}_i),y_{\theta t}-\tilde{x}_t\rangle$\\

$=\dot{\theta}_t\langle f(y_{\theta t})-f(\tilde{x}_t),y_{\theta t}-\tilde{x}_t\rangle$\\
\hspace*{\fill}$+\langle \dot{\theta}_tf(\tilde{x}_t)-f(\tilde{x}_i),y_{\theta t}-\tilde{x}_t\rangle$\\

$=\dot{\theta}_t\langle f(y_{\theta t})-f(\tilde{x}_t),y_{\theta t}-\tilde{x}_t\rangle$
$+\langle -f(\tilde{x}_i),y_{\theta t}-\tilde{x}_t\rangle$\\
%\hspace*{\fill}$-(1-\dot{\theta t})\langle f(\tilde{x}_i),y_{\theta t}-\tilde{x}_t\rangle$\\
%
\hspace*{\fill}(using the fact
$f(\tilde{x}_t)$ $\bot$ $(y_{\theta t}-\tilde{x}_t)$)\\

$\leq \dot{\theta}_t\Lambda_{i+1} \tilde{e}_{y_i}^2(t)$ 
$+\langle -f(\tilde{x}_i), y_{\theta t}-\tilde{x}_t\rangle$
%-(1-\dot{\theta t})\langle f(\tilde{x}_i),x_{\theta t}-\tilde{x}_t\rangle$\\
%
\hspace*{\fill}(using \eqref{eq:mu1})\\ %\footnote{possible because $\tilde{x}_t\in S$ and $x_{\theta t}\in D_{\varepsilon_i}(\tilde{x}_t)$ since $|x_{\theta t}-\tilde{x}_t|=\delta_t<\varepsilon_i$})\\
%
%\hspace*{\fill}and $\langle f(\tilde{x}_t)-f(x_{\theta t}),\tilde{x}_t-x_{\theta t}\rangle=0$ if $\tilde{x}_t\in\Delta$)\\

%$\leq \dot{\theta}_t\frac{\lambda_i}{2} \delta^2_t-\langle f(\tilde{x}_i),x_{\theta t}-\tilde{x}_t\rangle$\\
%
%(using the fact $f(\tilde{x}_t)$ $\bot$ $(x_{\theta t}-\tilde{x}_t)$)\\

$\leq\dot{\theta}_t\Lambda_{i+1}\tilde{e}_{y_i}^2(t)$
$+L|f(\tilde{x}_i)|t\tilde{e}_{y_i}(t)$\\
\hspace*{\fill}(because, using again 
$f(\tilde{x}_t)$ $\bot$ $(y_{\theta t}-\tilde{x}_t)$:

\hspace*{\fill}$|\langle -f(\tilde{x}_i), y_{\theta t}-\tilde{x}_t\rangle|=|\langle f(\tilde{x}_t)-f(\tilde{x}_i), y_{\theta t}-\tilde{x}_t\rangle|$

\hspace*{\fill}$\leq L|\tilde{x}_t-\tilde{x}_i||y_{\theta t}-\tilde{x}_t|
\leq Lt|f(\tilde{x}_i)| \tilde{e}_{y_i}(t)$\\
%\footnote{where $M^*$ is an upperbound of the sequence $\{f(\tilde{x}_i)\}_i$}).\\
Hence for all $t\in[0,h]$:
\begin{equation}\label{eq:key}
\frac{1}{2}\frac{d}{dt}(\tilde{e}_{y_i}^2(t))\leq \rho_{i+1}\Lambda_{i+1} \tilde{e}_{y_i}^2(t)+L|f(\tilde{x}_i)|t\tilde{e}_{y_i}(t)
\end{equation}
(with $\rho_{i+1}=a_{i+1}$ if $\Lambda_{i+1}<-\gamma$, $\rho_{i+1}=b_{i+1}$ otherwise).\\
%
%Hence:
%\begin{equation*} %\label{eq:LM}
%\tilde{e}_{y_i}(t)=|y_{\theta t}-\tilde{x}_t|\leq \delta_i(t).
%\end{equation*}
%with $\delta_i=|x_i-\tilde{x}_i|=d_\Delta(\tilde{x}_i)$. 
%The rest of the proof is exactly as in the corresponding part of the proof of Theorem~1 of~\cite{SNR17}.
%
%\begin{remark}
By subtracting $L|f(\tilde{x}_i)|t\tilde{e}_{y_i}(t)$ from $\rho_{i+1}\Lambda_{i+1}\tilde{e}_{y_i}^2(t)$
in the right-hand side of \eqref{eq:key}, we get
using \eqref{eq:minsigma} %and Definition~\ref{def:Dgamma}
%\footnote{suppose : $\tilde{e}_{y_i}(t)\neq 0$???}

%If~$\tilde{e}_{y_i}(t)\geq h\frac{4LM}{\gamma\min_{k\in\{0,1,\dots,i,i+1\}}a_k}=h\frac{1}{D'_{i+1}(\gamma)}$ for~all~$t\in[0,h]$
If 
%$h\leq \frac{\gamma}{2LM} a_{i}\tilde{e}_{y_i}$\footnote{cette condition est surement inutile, car quand $t=0$, le terme $LMt$ n'intervient pas???} and
$h\leq \frac{\gamma a_{i+1}}{2L|f(\tilde{x}_i)|} \tilde{e}_{y_i}(t)$ for~all~$t\in(0,h]$, then

%\begin{itemize}
%\item 
\hspace*{10mm} $\frac{1}{2}\frac{d}{dt}(\tilde{e}_{y_i}^2(t))\leq  \sigma_{i+1}\tilde{e}_{y_i}^2(t)$.\\
%
%\mbox{ if } \Lambda_i<-\gamma$, 
%
%\item $\frac{1}{2}\frac{d}{dt}(\tilde{e}_{y_i}^2(r))\leq \tilde{e}_{y_i}^2(r)\sigma_i \mbox{ if } \Lambda_i\geq -\gamma$.
%\end{itemize}
It then follows by integration %\footnote{utilise asussi un ``lemme'' classique de la litterature \`a mentionner???}  
that, for all 
$i\in\mathbb{N},y_i\in Y_i$: \\
%
%\begin{proposition}
If 
%$h\leq \frac{\gamma}{2LM} a_{i}\tilde{e}_{y_i}$ and
$h\leq \frac{\gamma a_{i+1}}{2L|f(\tilde{x}_i)|} \tilde{e}_{y_i}(s)$
for all $s\in[0,h]$, then,
assuming $\tilde{e}_{y_i}\leq \delta_i$:
%
%If~$h\leq \gamma D\min_{k\in\{i,i+1\}}a_k \inf_{s\in[0,h]}\tilde{e}_{y_i}(s)$, %and $\tilde{e}_{y_i}\leq\delta_{i}$ then
\begin{equation}\label{eq:key20}
%\tilde{e}_{y_i}(s)\leq \tilde{e}_{y_i}e^{\sigma_{i+1} s}=\delta_i(s) \mbox{ for all $s\in[0,h]$}.
\tilde{e}_{y_i}(s)\leq \delta_{i}e^{\sigma_{i+1} s}=\delta_i(s) \mbox{ for all $s\in[0,h]$}. %\footnote{on peut mettre: for all $s\in[0,h]$???}.
\end{equation}
Suppose that the hypothesis 
$h\leq \frac{\gamma a_{i+1}}{2L|f(\tilde{x}_i)|} \tilde{e}_{y_i}(s)$ is not satisfied
for some $s_0\in[0,h]$, i.e:
$\tilde{e}_{y_i}(s_0)<Ch$
where $C=\frac{2L|f(\tilde{x}_i)|}{\gamma a_{i+1}}$.
It then follows for all $s\in[0,h]$:\\
$\tilde{e}_{y_i}(s)<Ch+|\tilde{e}_{y_i}(s_0)-e_{y_i}(s)|$

$\leq Ch+|\xi_{y_i}(\theta_{y_i}(s))-\xi_{y_i}(\theta_{y_i}(s_0))|
+|\tilde{x}_i(s)-\tilde{x}_i(s_0)|$

$\leq Ch+hb_{i+1}M_f+h\tilde{M}_{i+1}$.\\
%where $M_{i+1}$ is an upperbound on magnitude of elements of $Z_{i+1}$ and $\tilde{M}_{i+1}$ is an upperbound of $|\tilde{x}_i(s)|$ for $s\in[0,h]$.
Using \eqref{eq:key20}, we have then for all $s\in[0,h]$:

$\tilde{e}_{y_i}(s)\leq\max(\delta_i(s),h\tilde{M}_{i+1}(\frac{2L}{\gamma a_{i+1}}+1)+b_{i+1}M_{f})$,
i.e.~\eqref{eq:key2}. %, which is equivalent to \eqref{eq:key2bis}.

By induction on~$i$,
using the fact that $\tilde{e}_{y_0}=|y_0-\tilde{x}_0|\leq \delta_0$
(since $y_0\in Y_0\subseteq B(\tilde{x}_0,\delta_0)$),
we prove \eqref{eq:max0}.
\end{proof}
\begin{definition}
For all $i\in\mathbb{N}$, let  
%${\cal C}_0=Y_0$ and
%$i\in[N_1]$ with $N_1=\lceil \tilde{R}_1/h \rceil$, let us define:
$${\cal C}_{i+1}=\bigcup_{s\in[0,h]}B(\tilde{x}_{i}(s),\delta_{i}(s))\cap S_{i}(s).$$
Besides
$${\cal C}=\bigcup_{i=1}^{N_1}{\cal C}_i,$$
where $N_1=\lceil\frac{\tilde{R}_1}{h}\rceil$ and
$\tilde{R}_1\in\mathbb{R}_{>0}$ is the 1st return time
of $\tilde{x}(t)$ to $S_0$.
\end{definition}
\noindent
%We now give a condition (see \eqref{eq:new}) that
%guarantees the existence of a limit cycle $\Gamma$, and 
%highlights a set ${\cal C}$ containg all solutions of \eqref{eqn:dyn}
%starting from $Y_0$.
%
Using Proposition~\ref{prop:basicn}, we 
now show that a simple test of set inclusion 
(see \eqref{eq:new})
guarantees
the existence of a limit cycle.
\begin{theorem}\label{lemma:new1}
Consider an Euler trajectory $\tilde{\Gamma}$
and $\delta_0\in \mathbb{R}_{>0}$. Suppose
that, for all $i\in[N_1],s\in[0,h]$:
\begin{equation}\label{eq:h}
%\delta(t)\geq h(\frac{2LM^*}{\gamma \min_{i\in[N_1]}a_i}+M+M^*)
\delta_{i-1}(s)\geq h(\tilde{M}_i(\frac{2L}{\gamma a_{i}}+1)+b_iM_{f}),
%\footnote{meme chose: regarde si ce n'est pas $i\in\{0,1,\dots,N_1\}$ et si ce n'est pas $a_{i+1}$???} \ \ \mbox{ for all $t\in[0,N_1h]$}
\end{equation}
and
\begin{equation}\label{eq:new}
{\cal C}_{N_1}\cap S_0\subseteq Y_0.
\end{equation} 
Suppose furthermore that there exists $\eta\in\mathbb{R}_{>0}$ such that
\eqref{eq:eta} holds.
Then
\begin{itemize}
\item ${\cal C}$  is forward invariant w.r.t. $Y_0$, i.e.: For all $t\in[0,\infty),y\in Y_0$, %\footnote{or $y\in B(\tilde{x}_0,\delta_0)\cap S_0$???}, 
we have $\xi_{y}(t)\in{\cal C}$, and
\item ${\cal C}$ contains a limit cycle $\Gamma$.
\end{itemize}
\end{theorem}
\begin{proof}
%\subsection*{Proof  of the forward invariance of ${\cal C}$ and existence of $\Gamma\subset {\cal C}$}
%
From \eqref{eq:max0} and \eqref{eq:h}, we deduce:
%from Proposition~\ref{prop:basicn}, 
$e_{y_0}(t)\leq \delta(t)$ for all $y_0\in Y_0, t\in[0,N_1h]$.
In particular, at the time of the 1st return of $\tilde{x}(t)$
to $S_0$ at $t=\tilde{R}_1\in ((N_1-1)h,N_1h]$,
we have $e_{y_0}(\tilde{R}_1)\leq \delta(\tilde{R}_1)$.
So: $e_{y_0}(\tilde{R}_1)=|y_1-\tilde{x}_1|\leq \delta(\tilde{R}_1)$, where $\tilde{x}_1=\tilde{x}(\tilde{R}_1)$ and $y_1=\xi_{y_0}(T_{y_0}(1))$
are the 1st return points of $\tilde{x}(t)$ and $\xi_{y_0}(t)$
to $S_0$ respectively.

From the definition of ${\cal C}_{N_1}$ we can deduce that
$y_1\in{\cal C}_{N_1}\cap S_0$. Therefore, using \eqref{eq:new}, we have: $y_1\in Y_0$.
Thus for all $y\in Y_0$, we have
$y'=\xi_{y_0}(T_{y}(1))\in Y_0$.
This shows by the Brouwer fixed-point theorem applied to the continuous function: $y\in Y_0\mapsto \xi_{y}(T_y(1))$, that there exists
$y^*\in {\cal C}_{N_1}\cap S_0\subseteq Y_0$ such that
$\xi_{y^*}(T_{y^*}(1))=y^*$.
%(car $N_1=\lceil \frac{T_{y^*}{{h} \rceil$ and $T_1(y^*)=T\in[(N_1-1)h,N_1h]$).
%
We deduce that there exists a limit cycle $\Gamma$ (closed solution
de \eqref{eqn:dyn}) of period $T=T_{y^*}(1)$
which passes through $y^*\in{\cal C}_{N_1}\cap S_0\subseteq Y_0$.

Observe that for all $y_0\in Y_0$ and $t\in[0,T_{y_0}(1)]$, we have,
by \eqref{eq:max0} of Proposition~\ref{prop:basicn},
$e_{y_k}(s)=|\tilde{x_k}(s)-\xi(\theta_{y_k}(s))|\leq \delta_k(s)$,
hence $\xi(\theta_{y_k}(s))\in B(\tilde{x}_k(s),\delta_k(s))\cap S_k(s)\in
{\cal C}_k$, and

$\xi_{y_0}(t)\in \bigcup_{k=1}^{N_1} {\cal C}_k={\cal C}$ for
all $t\in[0,T_{y_0}(1)]$.\\
Since $y_1=\xi_{y_0}(T_{y_0}(1))\in{\cal C}_{N_1}\cap S_0\subseteq Y_0$,
we can repeat the reasoning starting from $y_1$, which gives
$\xi_{y_1}(t)\in{\cal C}$ for $t\in[0,T_{y_1}(1)]$,
hence $\xi_{y_0}(t)\in{\cal C}$ for $t\in[0,T_{y_0}(2)]$.
By induction on $p\in\mathbb{N}_{>0}$, we have more generally
\begin{equation}\label{eq:Tp}
\xi_{y_0}(t)\in{\cal C} \mbox{ for all $t\in[0,T_{y_0}(p)]$}.
\end{equation}
Now, using \eqref{eq:eta} and~\eqref{eq:new}, we have:
$T_{y_0}(p)\geq p\eta$ for all $p\in\mathbb{N}_{>0}$, hence
$T_{y_0}(p)\rightarrow\infty$ as $p\rightarrow\infty$. 
So we deduce from \eqref{eq:Tp}:
$\xi_{y_0}(t)\in{\cal C}$ for all $t\in[0,\infty)$.
In particular, the limit cycle $\Gamma$ which is the trajectory
$\xi_{y^*}(t)$ for $t\in[0,\infty)$, is included in~${\cal C}$.
\end{proof}
Note that a simple numerical sufficient condition for $\eqref{eq:new}$ is:
$|\tilde{x}(\tilde{R}_1)-\tilde{x}_0|+\delta(\tilde{R}_1)<\delta_0$. %\footnote{a verifier???}
%
%Thiswill be used in Proposition~\ref{prop:basico}. %\footnote{Attention, ca se mord la queue entre Prop 1 et Theorem 1???}
%
\begin{example}\label{ex:1}
Consider the 
Van der Pol system with $p=0.3$.
\begin{equation*}
\begin{cases}
du_1/dt=u_2\\
du_2/dt=p u_2 -p u_1^2u_2-u_1\\
\end{cases}
\end{equation*}
We use the software ORBITADOR %~\cite{Jerray22}
to perform the simulations
(see \url{https://perso.telecom-paristech.fr/jjerray/orbitador/}).
Let $h=10^{-4}, \delta_0=0.1$,\\
$\tilde{x}_0=(1.8929; -0.5383),\gamma=0.015$.
We have \\
$L=1.516, M_f=2.22, \max_{i\in[N_1]}\tilde{M}_i=2.22$,
$\tilde{R}_1=6.314$,
$N_1=\lceil \frac{\tilde{R}_1}{h}\rceil=63140$,
and $\min_{k\in[N_1]}a_k\geq 0.9$, $\max_{k\in[N_1]}b_k\leq 1.1$. \\
Besides for all $t\in[0,N_1h]$: $\delta(t)\geq \delta(N_1h)=0.0642$.\\
Note that $\mu^\bot[J(\tilde{x}(t))]\in [-4.2,+1.2]$ for $t\in[0,\tilde{R}_1]$,
so the system is {\em not} uniformly contractive.\\
We check that \eqref{eq:h} holds, i.e, for all $i\in[N_1],s\in[0,h]$:

$\delta_{i-1}(s)\geq 0.0642\geq
h\max_{i\in[N_1]}(\tilde{M}_{i}(\frac{2L}{\gamma a_{i}}+1)+b_iM_{f})\approx 0.05$.\\
Besides, as checked by ORBITADOR, the inclusion
\eqref{eq:new} also holds. 
Hence by Theorem~\ref{lemma:new1}, there is a limit cycle $\Gamma$
contained in an invariant set ${\cal C}$.
The inclusion \eqref{eq:new} is highlighted in the following figures,
using the time-step $h=0.002$ (instead of $h=10^{-4}$)
in order to make the inclusion more visible.
On Figure~\ref{fig:v_cycle_0},
the initial ball $B(\tilde{x}_0,\delta_0)$ is
depicted in orange,
the invariant zone ${\cal C}$ in red,  
and 
the set $\bigcup_{s\in[0,h]}B(\tilde{x}_{N_1-1}(s),\delta_{N_1-1}(s))$ in green.
The sets 
${\cal C}_{N_1}\cap S_{N_1-1}$, $Y_0$ and ${\cal C}_{N_1}\cap S_{N_1}$ have the form
of straight line segments which are transverse to the 
Euler trajectory $\tilde{\Gamma}$.
On Figure~\ref{fig:v_cycle_0},
$\tilde{\Gamma}$ is depicted in black:
it starts at the middle $\tilde{x}_0$
of segment $Y_0$ (in blue),
and turns clockwise.
The segments ${\cal C}_{N_1}\cap S_{N_1-1}$, $Y_0$ and ${\cal C}_{N_1}\cap S_{N_1}$ 
appear successively from top to bottom 
in Figures~\ref{fig:v_cycle_0}
and~\ref{fig:v_cycle}.
(The bottom of Figure~\ref{fig:v_cycle_0}
and Figure~\ref{fig:v_cycle}
gives zoom views on ${\cal C}_{N_1}\cap S_{N_1-1},Y_0,{\cal C}_{N_1}\cap S_{N_1}$.)
We see \eqref{eq:new}:
${\cal C}_{N_1}\cap S_0\subseteq B(\tilde{x}_0,\delta_0)$,
where ${\cal C}_{N_1}$ is the part of the green zone delimited by 
${\cal C}_{N_1}\cap S_{N_1-1}$ (fuschia) and 
${\cal C}_{N_1}\cap S_{N_1}$ (cyan), 
$B(\tilde{x}_0,\delta_0)$ is the orange ball,
and~$S_0$ is the straight line extending $Y_0$ (blue).
%
%\begin{remark}
\begin{figure}[h!]
\centering
\includegraphics[scale=0.5]{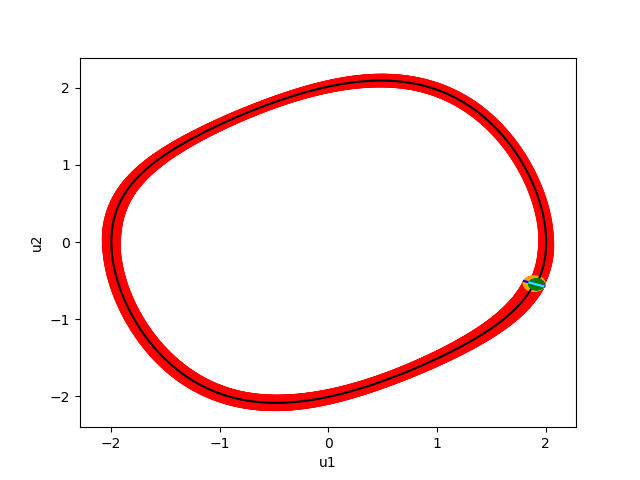}
\includegraphics[scale=0.5]{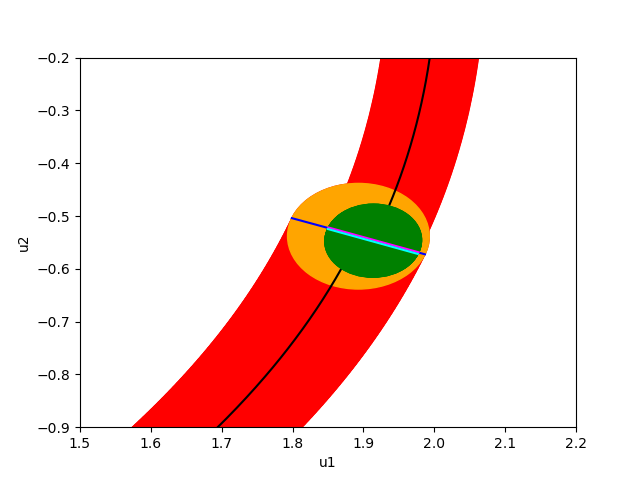}
\caption{Top: Euler trajectory $\tilde{\Gamma}$ (black); 
invariant zone ${\cal C}$  (red); $\bigcup_{s\in[0,h]}B(\tilde{x}_{N_1-1}(s),\delta_{N_1-1}(s))$
 (green); 
segments 
${\cal C}_{N_1}\cap S_{N_1-1}$
(fuschia),
$Y_0$ (blue),
and ${\cal C}_{N_1}\cap S_{N_1}$ (cyan).
Bottom: zoom showing \eqref{eq:new}:
${\cal C}_{N_1}\cap S_0\subseteq B(\tilde{x}_0,\delta_0)$,
where ${\cal C}_{N_1}$ is the green part delimited by ${\cal C}_{N_1}\cap S_{N_1-1}$ and 
${\cal C}_{N_1}\cap S_{N_1}$, and $B(\tilde{x}_0,\delta_0)$ is the orange ball.}
%$\lambda_{\min}(H(t))$
\label{fig:v_cycle_0}
\end{figure}
\begin{figure}[h!]
\centering
\includegraphics[scale=0.5]{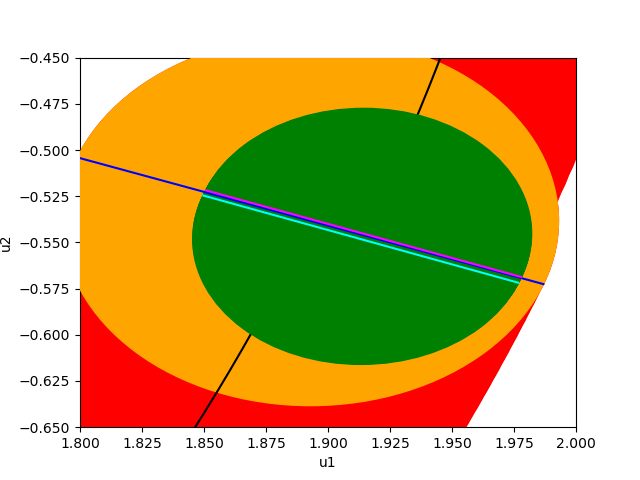}
\includegraphics[scale=0.5]{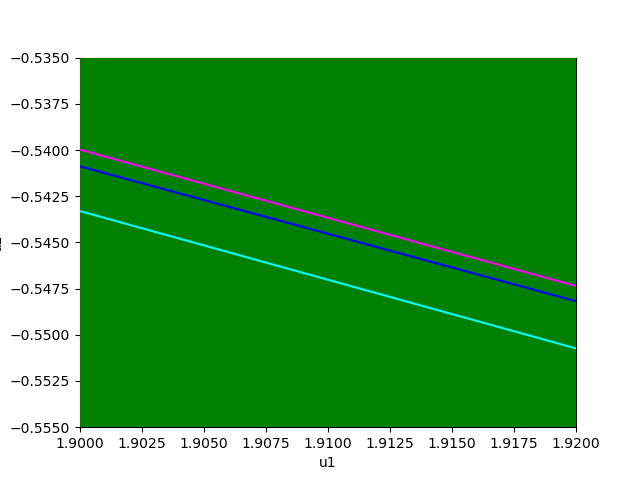}
\caption{Further zooms on ${\cal C}_{N_1}$ showing from top to bottom 
${\cal C}_{N_1}\cap S_{N_1-1}$
(fuschia), $Y_0$ (blue) 
and ${\cal C}_{N_1}\cap S_{N_1}$ 
(cyan). We see \eqref{eq:new}:
${\cal C}_{N_1}\cap S_0\subseteq B(\tilde{x}_0,\delta_0)$.
}
\label{fig:v_cycle}
%$\lambda_{\min}(H(t))$
\end{figure}
\end{example}
\section{Attractivity of the limit cycle}\label{sec:attraction}

We know assume that the existence of the limit cycle $\Gamma$ is known.
Let $y^*=\Gamma\cap S_0$ be the intersection of $\Gamma$
with $S_0$.

So far, the Euler approximation function $\tilde{x}(t)$ has been defined implicitly
with $\tilde{x}_0$ as initial condition. Let us now write
$\tilde{x}(t;z)$ to denote the Euler approximation function with $z$ as
initial condition, and let us write $\sigma_k(z)$ the corresponding
coefficient $\sigma_k$.
For $p\in\mathbb{N}_{>0}$ and $z\in Y_0$,
let also  $\tilde{R}_p(z)$ be
the time $t$ taken by $\tilde{x}(t;z)$ to make $p$ rounds,
and $N_p(z)=\lceil \frac{R_p(z)}{h}\rceil$.
Note that under the hypotheses of Theorem~\ref{lemma:new1}:
$\tilde{x}_p=\tilde{x}(\tilde{R}_p(\tilde{x}_0);\tilde{x}_0)
=\tilde{x}(\tilde{R}_1(\tilde{x}_{p-1});\tilde{x}_{p-1})\in Y_0$
for all $p\in\mathbb{N}_{>0}$.

Let us suppose that 
there exists $a,b\in \mathbb{R}_{>0}$ such that
\begin{equation}\label{eq:a}
0<a\leq a_i\leq b_i\leq b \mbox{ for all $i\in\mathbb{N}$}.
\end{equation}
This property is guaranteed if we know that all round performed by the
solution of \eqref{eqn:dyn} is completed in a known interval of time, i.e. if
there exists $T',T''\in\mathbb{R}_{>0}$ such that
\begin{equation}\label{eq:bornesup}
T'\leq T_1(y_0)\leq T'' \mbox{ for all $y_0\in Y_0$}.
\end{equation}
We will also use the assumption:
there exists $R'\in\mathbb{R}_{>0}$ such that
\begin{equation}\label{eq:bornesupbis}
\tilde{R}_1(z)\leq R' \mbox{ for all $z\in Y_0$}.
\end{equation}

\begin{theorem}\label{th:key1} 
Consider an Euler trajectory $\tilde{\Gamma}$
and $\delta_0\in \mathbb{R}_{>0}$. Suppose
that the system satisfies \eqref{eq:h} and \eqref{eq:new}.
Suppose furthermore that there exists $\eta,T',T'',R'\in\mathbb{R}_{>0}$ such that
\eqref{eq:eta}, \eqref{eq:bornesup} and~\eqref{eq:bornesupbis} hold.
Then
we have %(by \eqref{eq:max0} of Proposition~\ref{prop:basicn}) 
for all $t\in[0,\infty), y\in Y_0$:
\begin{equation}\label{eq:maxnew}
\tilde{e}_{y}(t)\leq \max(\delta(t), D h),
\end{equation}
where
$D=M_{{\cal C}}(\frac{2L}{\gamma a}+b+1)$ and
$M_{{\cal C}}$ is an upper bound on magnitude of $f$ over the elements of ${\cal C}$.
Besides if 
\begin{equation}\label{eq:limsup}
\lim_{t\rightarrow\infty}\delta(t)=0,
\end{equation}
then 
\begin{equation}\label{eq:limsup2}
\limsup_{t\rightarrow\infty}\tilde{e}_{y}(t)\leq Dh,
 \mbox{ for all $y\in Y_0$},
\end{equation}
%It then follows:
\begin{equation}\label{eq:Gamma1}
\limsup_{t\rightarrow\infty} d(\tilde{x}(t),\Gamma) \leq D h
\end{equation}
(hence, the Euler trajectory $\tilde{x}(t;\tilde{x}_0)$ converges asymptotically to~$\Gamma$ as $h\rightarrow 0$)
and:
\begin{equation}\label{eq:Gamma2}
\lim_{t\rightarrow\infty} d(\xi_{y}(t),\Gamma)=0
\mbox{ for all $y\in Y_0$},
\end{equation}
i.e., $Y_0$ is a basin of attraction of $\Gamma$.\\
A sufficient condition for 
\eqref{eq:limsup} to hold is:

There exists
%$U\in\mathbb{R}_{>0}$ such that
%\begin{equation}\label{eq:bornesup}
%N_1(z)\leq U \mbox{ for all $z\in Y_0$},
%\end{equation}
$d\in\mathbb{R}_{<0}$ such that
\begin{equation}\label{eq:newrelax}
h\Sigma_{k=1}^{N_1(z)-1}\sigma_k(z)
+s\sigma_{N_1(z)}(z) \leq d<0
\end{equation}
for all $z\in Y_0, s\in[0,h]$.
\end{theorem}
\begin{proof}
By Theorem \ref{lemma:new1}, ${\cal C}$ is an invariant, and $M_{{\cal C}}$ can be taken
in place of $M_f$ and $\tilde{M}_i$ for all $i\in\mathbb{N}$.
From \eqref{eq:max0} of Proposition~\ref{prop:basicn}
and \eqref{eq:a} (which is implied by \eqref{eq:bornesup}), we then deduce:
$\tilde{e}_{y}(t)\leq \max(\delta(t),
M_{{\cal C}}h(\frac{2L}{\gamma a}+b+1))=\max(\delta(t),Dh)$,
i.e. \eqref{eq:maxnew}.
So if $\lim_{t\rightarrow\infty}\delta(t)=0$, $\delta(t)$ becomes
after a sufficient long time $t\geq t_0$ less than $Dh$, and 
by \eqref{eq:maxnew}, we have: $\tilde{e}_y(t)\leq Dh$ for $t\geq t_0$,
i.e. \eqref{eq:limsup2}. In particular, since $y^*\in Y_0$, for $y=y^*$,
we have $\limsup_{t\rightarrow\infty}\tilde{e}_{y^*}(t)\leq Dh$,
i.e. $\limsup_{t\rightarrow\infty}|\tilde{x}(t)-\xi_{y^*}(\Theta_{y^*}(t))|\leq Dh$,
hence $\limsup_{t\rightarrow\infty} d(\tilde{x}(t),\Gamma)\leq D$,
i.e. \eqref{eq:Gamma1},
since $\xi_{y^*}(t)\in\Gamma$ for all $t$.
Besides we have: 

$d(\xi_y(\Theta_y(t)),\Gamma)$

$\leq d(\xi_y(\Theta_y(t)),\tilde{x}(t))+d(\tilde{x}(t),\Gamma)$

$\leq e_{y}(t)+d(\tilde{x}(t),\Gamma)$.\\
By taking the $\limsup$ of both sides of the inequality, we have:

$\limsup_{t\rightarrow\infty}d(\xi_y(\Theta_y(t)),\Gamma)$

$\leq \limsup_{t\rightarrow\infty} \tilde{e}_{y}(t)+\limsup_{t\rightarrow\infty}d(\tilde{x}(t),\Gamma)$,\\
hence, using \eqref{eq:limsup2} and \eqref{eq:Gamma1}:

$\limsup_{t\rightarrow\infty}d(\xi_y(\Theta_y(t)),\Gamma)
\leq 2 Dh$, \\
and
since $\Theta_y(t)\rightarrow\infty$ as $t\rightarrow\infty$:

$\limsup_{t\rightarrow\infty}d(\xi_y(t),\Gamma)
\leq 2 Dh$.\\
Now since $d(\xi_y(t),\Gamma)$ does not depend on $h$, we have:
$\lim_{t\rightarrow\infty}d(\xi_y(t),\Gamma)=0$, i.e.
\eqref{eq:Gamma2}.\\
Let us now prove that \eqref{eq:newrelax} implies \eqref{eq:limsup}.
We first show by induction on $p$
that for all $\tilde{x}_0\in Y_0$
\begin{equation}\label{eq:induction}
\delta(\tilde{R}_1(\tilde{x}_{p-1}))\leq \delta_0e^{pd} \mbox{ for all $p\in\mathbb{N}_{>0}$}.
\end{equation}
For the base case $p=1$, we have for some $s\in(0,h]$\\
$\delta(\tilde{R}_1)\leq\delta_0 e^{h\Sigma_{k=1}^{N_1-1}\sigma_k+s\sigma_{N_1}}\leq \delta_0 e^d$ \hspace*{\fill} (using \eqref{eq:newrelax}).

For $p\geq 2$, we have for some $s\in(0,h]$\\
$\delta(\tilde{R}_p)\leq\delta(R_{p-1}) e^{h\Sigma_{k=1}^{N_1(\tilde{x}_{p-1})-1}\sigma_k(\tilde{x}_{p-1})+s\sigma_{N_1(\tilde{x}_{p-1})}(\tilde{x}_{p-1})}$

$\leq \delta(R_{p-1}) e^d$ 
\hspace*{\fill}(using \eqref{eq:newrelax} and the fact that $\tilde{x}_{p-1}\in Y_0$)

$\leq e^{(p-1)d}e^d=e^{pd}$ \hspace*{\fill} (by induction hypothesis).

From \eqref{eq:induction} it then follows that
$\delta(t)\leq \delta_0 e^{pd+B}$ for $t\in[\tilde{R}_p,\tilde{R}_{p+1}]$,
where 

$B=h\max_{i\in[N_1(\tilde{x}_p)-1],s\in\{0,s\}}\Sigma_{k=1}^ih\sigma_k(\tilde{x}_p)+s\sigma_{i+1}(\tilde{x}_p)$.\\
Using the fact that $d<0$ and 
that under \eqref{eq:bornesupbis}, $p\rightarrow\infty$ as  $t\rightarrow\infty$,
we have:
$\lim_{t\rightarrow\infty}\delta(t)=0$, i.e. \eqref{eq:limsup}.
\end{proof}
\begin{remark}\label{rk:CSk}
Using a continuity argument, it can be seen that \eqref{eq:newrelax} 
holds for all $z\in Y_0, s\in[0,h]$, if
\begin{equation}\label{eq:newrelaxter}
\int_0^T\rho(t)\mu^\bot[J(\xi_{y^*}(t)]dt \leq 4d
\end{equation}
holds (with $\rho(t)=\frac{1}{2}$ if $\mu^\bot[J(\xi_{y^*}(t)]<\gamma$,
and $\rho(t)=\frac{3}{2}$ otherwise),
and if $h$ is sufficiently small, and
$\delta_0$ (hence $|\tilde{x}_0-y^*|$) is sufficiently small.
Note that  \eqref{eq:newrelaxter} is a strong form of \eqref{eq:periodic}:
$\int_0^T \mu^\bot[J(\xi_{y^*}(t))]dt \leq d<0$
(i.e., \eqref{eq:newrelaxter} implies~\eqref{eq:periodic}).
\end{remark}
As a recapitulation, it follows from Theorems~\ref{lemma:new1}
and~\ref{th:key1}:
%and Proposition~\ref{prop:preorder}: 
\begin{itemize}
\item \eqref{eq:h}-\eqref{eq:new} coupled with \eqref{eq:eta}
forms a sufficient condition
that guarantees the existence of a limit cycle $\Gamma$,
and allows us to construct an invariant zone ${\cal C}$ around $\Gamma$; 
\item \eqref{eq:newrelax} (coupled with \eqref{eq:bornesup}-\eqref{eq:bornesupbis})
%\footnote{and \eqref{eq:dl}???} 
is an additional 
numeric condition that allows us to determine
a basin of attraction of~$\Gamma$.
%whenever $\tilde{x}_0$ is sufficiently close to $\Gamma$ and $\delta'_0$ is sufficiently small for \eqref{eq:continuity} to be satisfied. %\eqref{eq:A} and \eqref{eq:d0}
\end{itemize}
Note that although \eqref{eq:newrelax} is a criterion that has to be {\em a priori} checked numerically,
it can be formally established using, e.g., interval arithmetic \cite{Moore1979}.
%Note that  is the counterpart of \eqref{eq:periodic}. 
\begin{example}\label{ex:2}
Let us continue the Van der Pol of example~\ref{ex:1}
(with $h=10^{-4}$). 
We have $a=0.9, b=1.1$ and  
$D=M_{{\cal C}}h(\frac{2L}{\gamma a}+b+1)\approx 500$.
%
%We have  $\tilde{R}_1=6.314$, $N_1=\lceil \frac{\tilde{R}_1}{h}\rceil=63140$.
%
The numeric computation of 
\begin{equation*} %\label{eq:newrelax}
{\cal K}(z,s): =h\Sigma_{k=1}^{N_1(z)-1}\sigma_k(z)
+s\sigma_{N_1(z)}(z)
\end{equation*}
for $z\in Y_0$, and $s\in[0,h]$ gives:
${\cal K}(z,s)\in(-0.34,-0.36)$, hence ${\cal K}(z,s)< d$
with $d=-0.34$.
Therefore \eqref{eq:newrelax} holds, hence \eqref{eq:limsup}.
It then follows from Theorem~\ref{th:key1} that 
\eqref{eq:limsup2} holds and
$Y_0=B(\tilde{x}_0,\delta_0)\cap S_0$ is a basin of attraction of~$\Gamma$.

Finally, we give a numerical evidence of \eqref{eq:limsup2}.
Figure~\ref{fig:w1-w2-cycle}
%\footnote{mettre $\tilde{e}$ et non $\delta$  en ordonn\''e sur la figure} 
gives the evolution of 
$\tilde{e}_{y_0}(t)=|\xi_{y_0}(\Theta_{y_0}(t))-\tilde{x}(t)|$
for $y_0=(1.8037; -0.5057)\in Y_0$ and different values of time-step $h$
(The exact solution $\xi_{y_0}(t)$ is here
approximated as $\tilde{x}(t;y_0)$ with time-step $10^{-6}$).
The different curves  $\tilde{e}_{y_0}(t)$ are given from top to bottom
%$\delta^{(2)}(t),\delta^{(3)}(t),\delta^{(4)}(t)$ with   
for  
$h_1=5\ 10^{-4}$ (red), $h_2=2.5\ 10^{-4}$ (green),
$h_3=1.25\ 10^{-4}$~(blue). The curves are in agreement with~\eqref{eq:limsup2}:
$\limsup_{t\rightarrow\infty}\tilde{e}_{y_0}(t)\leq D h$.
%$\lim\inf_{t\rightarrow\infty} \delta^{(k)}(t)>\lim\inf_{t\rightarrow\infty} \delta^{(k+1)}(t)$, and 
%
%
%\begin{figure}[h!]
%\centering
%\includegraphics[scale=0.5]{figures/w-cycle.png}
%\includegraphics[scale=0.5]{figures/w1_w2_alpha=0,01.png}
%\caption{Evolution of $|\tilde{w}|$ for GDE with
%$\alpha=0.01$.}
%$\lambda_{\min}(H(t))$
%\label{fig:w_cycle}
%\end{figure}
\begin{figure}[h!]
\centering
\includegraphics[scale=0.5]{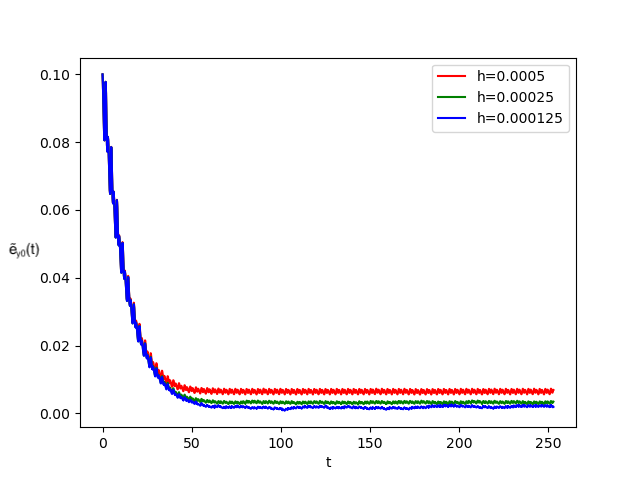}
%\includegraphics[scale=0.5]{figures/w1_w2_alpha=0,01.png}
%\caption{.}
%$\lambda_{\min}(H(t))$
%\label{fig:w1-w2-cycle}
%\end{figure}
%\begin{figure}[h!]
%\centering
%\includegraphics[scale=0.5]{figures/w1_w2_alpha=0,01.png}
\caption{Evolution of 
$\tilde{e}_{y_0}(t)$ 
with  $h_1=5\cdot 10^{-4}$ (red), $h_2=2.5\cdot 10^{-4}$ (green)
and $h_3=1.25\cdot 10^{-4}$ (blue),
showing agreement with
\eqref{eq:limsup2}:
$\limsup_{t\rightarrow\infty}\tilde{e}_{y_0}(t)\leq D h$.}
%$\lambda_{\min}(H(t))$
\label{fig:w1-w2-cycle}
\end{figure}
%\footnote{Construction de l'invariant entourant $\Gamma$ a l'aide de la methode de l'appendix ???}
\end{example}\label{ex:2}
\section{Final Remarks}\label{sec:conclusion}
%
%role de $\gamma$ in $\gamma D$??? 
%Si $\gamma$ small, $e(t)$ arrete de decroitre exponentiellment plus tot, mais si $\gamma$ small, le facteur de decroissance exponentiel $|\sigma|$ est plus faible.
%$\sigma_0, a_0, b_0, \Lambda_0$ ne servent \`a rien ???

We have given an original set of conditions 
involving the Euler approximation of \eqref{eqn:dyn}
that allows us to prove the existence of a limit cycle $\Gamma$,
and determine an invariant set ${\cal C}$ around $\Gamma$ as well as
a basin of attraction of $\Gamma$. 
%Dire que $\delta(t)$ is piecewise exponential, and $\delta(t)\rightarrow 0$ (for $h$ and $d(\tilde{x}_0,\Gamma)$ sufficiently small) if the integral of $\rho(t)\mu^\bot[(J(x(t))]$  on $\Gamma$ is decreasing  over a period $T$???

In the future, it would be interesting to apply this method 
together with non-Euclidean norms such as
the weighted norms of \cite{MaidensA15}. 
%
%Sur VdP, $\mu^\bot[J(z)]$ not uniformly $<0$.
%\eqref{eq:d0} is a strong form of \eqref{eq:periodic} if $d(\tilde{x}_0,\Gamma)\leq \delta'_0$)
\newcommand{\LNCS}{LNCS}

\ifdefined\VersionAuthor
	\renewcommand*{\bibfont}{\small}
	\printbibliography[title={References}]
\else
	\bibliographystyle{IEEEtran} % plain
	\bibliography{acc21bis}

% Generated by IEEEtran.bst, version: 1.14 (2015/08/26)
\begin{thebibliography}{1}
\providecommand{\url}[1]{#1}
\csname url@samestyle\endcsname
\providecommand{\newblock}{\relax}
\providecommand{\bibinfo}[2]{#2}
\providecommand{\BIBentrySTDinterwordspacing}{\spaceskip=0pt\relax}
\providecommand{\BIBentryALTinterwordstretchfactor}{4}
\providecommand{\BIBentryALTinterwordspacing}{\spaceskip=\fontdimen2\font plus
\BIBentryALTinterwordstretchfactor\fontdimen3\font minus
  \fontdimen4\font\relax}
\providecommand{\BIBforeignlanguage}[2]{{%
\expandafter\ifx\csname l@#1\endcsname\relax
\typeout{** WARNING: IEEEtran.bst: No hyphenation pattern has been}%
\typeout{** loaded for the language `#1'. Using the pattern for}%
\typeout{** the default language instead.}%
\else
\language=\csname l@#1\endcsname
\fi
#2}}
\providecommand{\BIBdecl}{\relax}
\BIBdecl

\bibitem{Soderlind06}
\BIBentryALTinterwordspacing
G.~S{\"o}derlind, ``The logarithmic norm. {H}istory and modern theory,''
  \emph{BIT Numerical Mathematics}, vol.~46, no.~3, pp. 631--652, Sep 2006.
  [Online]. Available: \url{https://doi.org/10.1007/s10543-006-0069-9}
\BIBentrySTDinterwordspacing

\bibitem{AminzareS14}
Z.~Aminzare and E.~D. Sontag, ``Contraction methods for nonlinear systems: {A}
  brief introduction and some open problems,'' in \emph{53rd {IEEE} Conference
  on Decision and Control, {CDC} 2014, Los Angeles, CA, USA, December 15-17,
  2014}, 2014, pp. 3835--3847.

\bibitem{Giesl2022}
P.~Giesl, S.~F. Hafstein, and C.~Kawan, ``Review on contraction analysis and
  computation of contraction metrics,'' \emph{Journal of Computational
  Dynamics}, 2022.

\bibitem{MaidensA15}
\BIBentryALTinterwordspacing
J.~N. Maidens and M.~Arcak, ``Reachability analysis of nonlinear systems using
  matrix measures,'' \emph{{IEEE} Trans. Autom. Control.}, vol.~60, no.~1, pp.
  265--270, 2015. [Online]. Available:
  \url{https://doi.org/10.1109/TAC.2014.2325635}
\BIBentrySTDinterwordspacing

\bibitem{FanM15}
\BIBentryALTinterwordspacing
C.~Fan and S.~Mitra, ``Bounded verification with on-the-fly discrepancy
  computation,'' in \emph{Automated Technology for Verification and Analysis -
  13th International Symposium, {ATVA}}, ser. Lecture Notes in Computer
  Science, vol. 9364.\hskip 1em plus 0.5em minus 0.4em\relax Springer, 2015,
  pp. 446--463. [Online]. Available:
  \url{https://doi.org/10.1007/978-3-319-24953-7\_32}
\BIBentrySTDinterwordspacing

\bibitem{ManchesterCDC13}
\BIBentryALTinterwordspacing
I.~R. Manchester and J.~E. Slotine, ``Transverse contraction criteria for
  existence, stability, and robustness of a limit cycle,'' in \emph{Proceedings
  of the 52nd {IEEE} Conference on Decision and Control, {CDC}}.\hskip 1em plus
  0.5em minus 0.4em\relax {IEEE}, 2013, pp. 5909--5914. [Online]. Available:
  \url{https://doi.org/10.1109/CDC.2013.6760821}
\BIBentrySTDinterwordspacing

\bibitem{LohmillerS98}
W.~Lohmiller and J.-J.~E. Slotine, ``On contraction analysis for non-linear
  systems,'' \emph{Automatica}, vol.~34, no.~6, pp. 683--696, 1998.

\bibitem{SNR17}
A.~Le~Coënt, F.~De~Vuyst, L.~Chamoin, and L.~Fribourg, ``Control synthesis of
  nonlinear sampled switched systems using {E}uler's method,'' in \emph{{SNR}},
  ser. {EPTCS}, vol. 247, 2017, pp. 18--33.

\bibitem{Moore1979}
\BIBentryALTinterwordspacing
R.~E. Moore, \emph{Methods and applications of interval analysis}, ser. {SIAM}
  studies in applied mathematics.\hskip 1em plus 0.5em minus 0.4em\relax
  {SIAM}, 1979. [Online]. Available:
  \url{https://doi.org/10.1137/1.9781611970906}
\BIBentrySTDinterwordspacing

\end{thebibliography}
\fi
%\newpage
%\section*{Appendix: Proof of Theorem \ref{th:key1}}
\end{document}